\documentclass[runningheads]{llncs}
\usepackage{multirow} 
\usepackage{microtype} 
\usepackage{xargs} 
\usepackage{ntheorem}
\usepackage{amssymb}
\usepackage{relsize}
\usepackage{hhline}
\usepackage{galois}
\usepackage{url}
\usepackage{tikz}
\usepackage{xspace}
\usepackage{multirow}
\usepackage{footnote}
\usepackage{stmaryrd}
\usepackage{pgfplots}
\usepackage{booktabs}
\makesavenoteenv{tabular}
\makesavenoteenv{table}
\usepackage{amsmath}
\usepackage{listings}
\newcommand{\prettylstformat}[0]{
\lstset{language=Prolog,
        numbers=left,numberstyle=\tiny,stepnumber=1,numbersep=8pt,
        frameround=tttt,
        frame=ltrb,
        basicstyle=\scriptsize\ttfamily,
        commentstyle=\color{gray},
        breaklines=true,breakatwhitespace=true,
        showlines=true,
        showspaces=false,showtabs=false,
        keywords={pred,prop},
        escapeinside=~~,
      }}

\usetikzlibrary{arrows,automata,positioning}
\usetikzlibrary{positioning}
\tikzset{main node/.style={fill=white,minimum size=0.5cm,inner sep=0pt},
}

\usepackage[textsize=footnotesize,prependcaption]{todonotes}

\newcommandx{\PLGout}[2][1=]{\todo[backgroundcolor=green,#1]{{\bf PLG:}~#2}}

\newcommandx{\MKout}[2][1=]{\todo[backgroundcolor=orange!20,#1]{{\bf MK:}~#2}}

\newcommand\costrelation{cost relation}
\newcommand\costrelations{cost relations}

\newcommand{\benchmarks}{\textbf{Bench}} 
\newcommand{\tabres}{\textbf{Res}} 
\newcommand{\colexp}{\textbf{Bound Inferred}}
\newcommand{\coltimes}{$\mathbf{T_A(ms)}$}
\newcommand{\seqcost}{\textbf{SCost}} 
\newcommand{\parcost}{\textbf{PCost}} 

\newcommand{\parthreads}{\textbf{SThreads}}

\newcommand{\timeunits}{ms} 
 
\newcommand{\atime}{\textbf{Time (\timeunits)}}

\newcommand{\bigo}{\textbf{BigO}}
\newcommand\mlength[1]{l_{#1}}
\newcommand\mint[1]{i_{#1}}

\newcommand{\headcost}{\varphi}

\newcommand\Prog{\mathcal P}

\newcommand\rs{{\cal R}_{\infty}}

\newcommand\cd{\Sigma}
\newcommand\ct{\vec \terme} 
\newcommand\cts{E}

\newcommand\szd{{\mathcal N}^{m}_{\top}}
\newcommand\sztypdom{\szd}

\newcommand\ruf[1]{{\mathcal C}_{#1}} 
\newcommand\eruf[1]{\hat{{\mathcal C}}_{#1}}

\renewcommand\vec\bar

\makeatletter
\def\hlinewd#1{
  \noalign{\ifnum0=`}\fi\hrule \@height #1 \futurelet
   \reserved@a\@xhline}
\makeatother

\newcommand\predp{\textup{\tt p}} 
\newcommand\predq{\textup{\tt q}}

\newcommand\terme{\textup{\tt e}}

\newcommand\varx{\textup{\tt x}}

\newcommand\vecx{\vec \varx}

\newcommand\neck{\mathrel{\textup{\tt :-}}}

\newcommand{\kbd}[1]{\mbox{\tt #1}}

\newcommand{\ciaopp}{CiaoPP\xspace}

\renewcommand\cd{{\Pi}}

\newcommand\resr{r}

\newcommand\stdub[2]{{\tt C}_{#1}(#2)}

\newcommand\beginequationspace{}
\newcommand\finishequationspace{}

\newcommand{\secbeg}{}
\newcommand{\secend}{}

\newcounter{mnotei}
\newcommand{\optionaltext}[1]{#1}
\renewcommand{\optionaltext}[1]{ }

\newtheorem*{ntheorem}{Theorem} 

\begin{document}
\label{firstpage}
\title{Towards a General Framework for Static Cost Analysis of Parallel Logic Programs}

\author{M. Klemen\inst{1,2} \and
P. L\'{o}pez-Garc\'{i}a\inst{1,3} \and
J.P. Gallagher\inst{1,4} \and \\
J.F. Morales\inst{1} \and
M.V. Hermenegildo\inst{1,2}
}

\authorrunning{F. Author et al.}

\institute{IMDEA Software Institute, Spain \and
  ETSI Inform\'{a}ticos, Universidad Polit\'{e}cnica de Madrid (UPM), Spain \and
  Spanish Council for Scientific Research (CSIC), Spain \and
  Roskilde University, Denmark\\ 
\email{\hspace*{-12mm}\{maximiliano.klemen,pedro.lopez,john.gallagher,josef.morales,manuel.hermenegildo\}@imdea.org} \\
}

\maketitle

\begin{abstract}
The estimation and control of \emph{resource usage} is now an
important challenge in an increasing number of computing systems. In
particular, requirements on timing and energy arise in a wide
variety of applications such as internet of things, cloud computing,
health, transportation, and robots.  At
the same time, parallel computing, with (heterogeneous) multi-core
platforms in particular, has become the dominant paradigm in computer
architecture. Predicting resource usage
on such platforms poses a difficult challenge.  Most work on static
resource analysis has focused on sequential programs, and relatively
little progress has been made on the analysis of parallel programs,
or more specifically on parallel logic programs.  We propose a novel, general, and
flexible framework for setting up cost equations/relations which can
be instantiated for performing resource usage analysis of parallel
logic programs for a wide range of resources, platforms and execution
models.  The analysis estimates both lower and upper bounds on the
resource usage of a parallel program (without executing it) as
functions on input data sizes. In addition, it also infers other
meaningful information to better exploit and assess the potential and
actual parallelism of a system. We develop a method for solving cost
relations involving the $max$ function that arise in the analysis of
parallel programs.
Finally, we instantiate our
general framework for the analysis of logic programs with Independent
And-Parallelism,
report on an implementation
within the \ciaopp{} system, and provide some experimental results.
To our knowledge, this is the first approach to the
cost analysis of \emph{parallel logic programs}.
\end{abstract}

\begin{keywords}
Resource Usage Analysis, Parallelism, Static Analysis, Complexity
Analysis, (Constraint) Logic Programming, Prolog.
\vspace*{-3mm}
\end{keywords}

\secbeg
\section{Introduction}
\label{sec:intro}
\secend

Estimating in advance the resource usage of computations is useful for
a number of applications; examples include granularity control in
parallel/distributed systems, automatic program optimization,
verification of resource-related specifications and detection of
performance bugs, as well as helping developers make resource-related
design decisions.
Besides \emph{time} and \emph{energy}, we assume a broad concept of
resources as numerical properties of the execution of a program,
including the number of \emph{execution steps}, the number of
\emph{calls} to a procedure, the number of \emph{network accesses},
number of \emph{transactions} in a database, and other user-definable
resources.
The goal of automatic static analysis is to estimate such properties
without running the program with concrete data, as a function of input
data sizes and possibly other (environmental) parameters.

Due to
the heat generation barrier in
traditional sequential architectures, parallel computing, with
(heterogeneous) multi-core processors in particular, has become the
dominant paradigm in current computer architecture.
Predicting resource usage on such platforms poses important
challenges.  Most work on static resource analysis
has focused on sequential programs, but relatively little progress has
been made on the analysis of parallel programs, or on parallel
logic programs in particular.
The significant body of work on static analysis of sequential logic
programs has already been applied to the analysis of other
programming paradigms, including imperative programs. This is achieved
via a transformation into \emph{Horn
  clauses}~\cite{decomp-oo-prolog-lopstr07-short}.
In this paper we concentrate on the analysis of parallel Horn clause
programs, which could be the result of such a translation from a
parallel imperative program or be themselves the source program.
Our starting point is the well-developed technique of setting up
recurrence relations representing resource usage functions
parameterized by input data
sizes~\cite{Wegbreit75-short-plus,Rosendahl89-short,granularity-short,caslog-short,low-bounds-ilps97-short,resource-iclp07-short,AlbertAGP11a-short,plai-resources-iclp14-short},
which are then solved to obtain (exact or safely approximated) closed
forms of such functions (i.e., functions that provide upper or lower
bounds on resource usage).
We build on this and propose a novel, general, and flexible framework
for setting up cost equations/relations which can be instantiated for
performing static resource usage analysis of parallel logic programs
for a wide range of resources, platforms and execution models.  Such
an analysis estimates both lower and upper bounds on the resource
usage of a parallel program
as functions on input data sizes.
We have instantiated the framework for dealing with
Independent And-Parallelism (IAP)~\cite{sinsi-jlp-short,partut-toplas-short},
which refers to the parallel execution of conjuncts in a goal.
However, the results can be applied to other languages and types of
parallelism, by performing suitable transformations into Horn clauses.

The main contributions of this paper can be summarized as follows:

\vspace*{-2mm}
\begin{itemize}
\item We have extended a general static analysis framework for the
  analysis of sequential Horn clause
  programs~\cite{resource-iclp07-short,plai-resources-iclp14-short},
  to deal with parallel programs.

\item Our extensions and generalizations support a wide range of
  resources, platforms and parallel/distributed execution models, and
  allow the inference of both lower and upper bounds on resource
  usage. This is
  the first approach, to our knowledge, to the cost analysis of \emph{parallel logic
    programs} that can deal with features such as
  backtracking, multiple solutions (i.e., non-determinism), and
  failure.
  
\item We have instantiated the developed framework to infer
  useful information for assessing and exploiting the potential
  and actual parallelism of a system.

\item We have developed a method for finding closed-form functions of
  cost relations involving the $max$ function that arise in the
  analysis of parallel programs.

\item We have developed a prototype implementation that instantiates
  the framework for the analysis of logic programs with Independent
  And-Parallelism within the \ciaopp{}
  system~\cite{ciaopp-sas03-journal-scp-short,resource-iclp07-short,plai-resources-iclp14-short},
  and provided some experimental results.

\end{itemize}

\secbeg
\section{Overview of the Approach}
\label{sec:Overview}
\secend

Prior to explaining our approach, we provide some preliminary
concepts.  Independent And-Parallelism
arises between two goals when their corresponding executions do not
affect each other. For pure goals (i.e., without side effects) a
sufficient condition for the correctness of IAP is the absence of
variable sharing at run-time among such goals.
IAP has traditionally been expressed using the \texttt{\&/2} meta-predicate
as the constructor to represent the parallel execution of goals. In
this way, the conjunction of goals (i.e., literals) \texttt{p \& q} in
the body of a clause will trigger the execution of goals \texttt{p}
and \texttt{q} in parallel, finishing when both executions
finish. 
        
Given a program $\Prog$ and a predicate $\predp \in \Prog$ of arity
$k$ and a set $\cd$ of $k$-tuples of calling data to $\predp$, we
refer to the \emph{(standard) cost} of a call $\predp(\ct)$ (i.e., a
call to $\predp$ with actual data $\vec e \in \cd$), as the resource
usage (under a given cost metric) of the complete execution of
$\predp(\ct)$. The \emph{standard cost} is formalized as a function
$\ruf{\predp}: \cd \rightarrow \rs$, where $\rs$ is the set of real
numbers augmented with the special symbol $\infty$ (which is used to
represent non-termination).
We extend the function $\ruf{\predp}$ to the powerset of $\cd$, i.e.,
$\eruf{\predp}: 2^{\cd} \rightarrow 2^{\rs}$, where
$\eruf{\predp}(\cts) = \{\ruf{\predp}(\ct) \mid \ct \in \cts\}$.  Our
goal is to abstract (safely approximate, as accurately as possible)
$\eruf{\predp}$ (note that $\ruf{\predp}(\ct) =
\eruf{\predp}(\{\ct\})$). Intuitively, this abstraction is the
composition of two abstractions: a size abstraction and a cost
abstraction. The goal of the analysis is to infer two functions
$\eruf{\predp}^{\downarrow}$ and $\eruf{\predp}^{\uparrow}: \sztypdom
\rightarrow \rs$ that give lower and upper bounds respectively on the
cost function $\eruf{\predp}$, where $\szd$ is the set of $m$-tuples
whose elements are natural numbers or the special symbol $\top$,
meaning that the size of a given term under a given size metric is
\emph{undefined}.  Such bounds are given as a function of tuples of
data sizes (representing the concrete tuples of data of the concrete
function $\eruf{\predp}$). Typical size metrics are the actual value
of a number, the length of a list, the size (number of constant and
function symbols) of a term,
etc.~\cite{resource-iclp07-short,plai-resources-iclp14-short}.

We now enumerate different metrics used to evaluate the performance of
parallel versions
of a logic program, compared against its corresponding sequential
version. When possible, we define these metrics parameterized with
respect to the resource (e.g., number of \emph{resolution steps},
\emph{execution time}, or \emph{energy consumption}) in which the cost
is expressed.

\begin{itemize}

\item \textbf{Sequential cost (\emph{Work})}: It is the standard cost
  of executing a program, assuming no parallelism.

\item \textbf{Parallel cost (\emph{Depth})}: It is the cost of
  executing a program in parallel, considering an unbounded number
  of processors.

\item \textbf{Maximum number of processes running in parallel
  ($Proc_{max}(P)$)}:
  The maximum number of processes that can run
  simultaneously in a program. This is useful, for example, to
  determine what is the minimum number of processors that are required
  to actually run all the processes in parallel.

\end{itemize}

The following example illustrates our approach.

\secbeg

\begin{example}
\label{examp:prime-numbers}

Consider the predicate \texttt{scalar/3} below, and a calling mode to
it with the first argument bound to an integer $n$ and the second one
bound to a list of integers $[ x_1, x_2, \cdots, x_k ]$. Upon success,
the third argument is bound to the list of products $[ n \cdot x_1, n
  \cdot x_2, \cdots, n \cdot x_k ]$. Each product is recursively
computed by predicate \texttt{mult/3}. The calling modes are
automatically inferred by
\ciaopp\ (see~\cite{ciaopp-sas03-journal-scp-short} and its
references): the first two arguments of both predicates are input, and
their last arguments are output.

\prettylstformat
\begin{lstlisting}[escapechar=\#]
scalar(_,[],[]).
scalar(N,[X|Xs],[Y|Ys]):- 
   mult(N,X,Y) & scalar(N,Xs,Ys).

mult(0,_,0).
mult(N,X,Y):-
   N>1,
   N1 is N - 1,
   mult(N1,X,Y0),
   Y is Y0 + X.
\end{lstlisting}

\noindent
The call to the parallel \texttt{\&/2} operator in the body of the
second clause of \texttt{scalar/3} causes the calls to \texttt{mult/3}
and \texttt{scalar/3} to be executed in parallel.  

We want to infer the cost of such a call to \texttt{scalar/3},
in terms of the number of resolution steps, as a function of its input
data sizes.
We use the \ciaopp{} system to infer size relations for the different
arguments in the clauses, as well as dealing with a rich set of size
metrics (see~\cite{resource-iclp07-short,plai-resources-iclp14-short}
for details). Assume that the size metrics used in this example are
the \emph{actual value} of \texttt{N} (denoted \texttt{int(N)}), for
the first argument, and the \emph{list-length} for the second and
third arguments (denoted \texttt{length(X)} and
\texttt{length(Y)}). Since size relations are obvious in this example,
we focus only on the setting up of \costrelations{} for the sake of
brevity.
Regarding the number of solutions, in this example all the predicates
generate at most one solution.
For simplicity we assume that
all builtin predicates, such as \texttt{is/2} and the comparison
operators have zero cost (in practice they have a ``trust''assertion
that specifies their cost as if it had been inferred by the system).
As the program contains parallel calls, we
are interested in inferring both total resolution steps, i.e.,
considering a sequential execution (represented by the \texttt{seq}
identifier), and the number of parallel steps, considering a parallel
execution, with infinite number of processors (represented by
\texttt{par}). In the latter case, the definition of this resource
establishes that the aggregator of the costs of the parallel calls that
are arguments of the \texttt{\&/2} meta-predicate is the \texttt{max/2}
function.  Thus, the number of parallel resolution steps for \texttt{p
  \& q} is the maximum between the parallel steps performed by
\texttt{p} and the ones performed by \texttt{q}. However, for
computing the \emph{total resolution steps}, the aggregation operator
we use is the addition, both for parallel and sequential calls. For
brevity, in this example we only
infer upper bounds on resource usages.

We now
set up the \costrelations{} for \texttt{scalar/3} and
\texttt{mult/3}. Note that the cost functions have two arguments,
corresponding to the sizes of the input arguments\footnote{For the
  sake of clarity, we abuse of notation in the examples when
  representing the cost functions that depend on data sizes.}.  In the
equations, we underline the operation applied as
cost aggregator for \texttt{\&/2}.

For the sequential execution (\texttt{seq}), we obtain the
following cost relations:

\vspace*{-1mm}
\beginequationspace
\begin{equation*}
    \begin{array}{rll}
      \stdub{\mathtt{scalar}}{n,l} = & 1 &  \text{ if } l = 0 \\ [-1mm]
      \stdub{\mathtt{scalar}}{n,l} = & \stdub{\mathtt{mult}}{n} \underline{\underline{+}} \stdub{\mathtt{scalar}}{n,l-1} + 1 & \text{ if } l > 0 \\
      \\ [-4mm]
      \stdub{\mathtt{mult}}{n} = & 1 &  \text{ if } n = 0 \\
      \stdub{\mathtt{mult}}{n} = & \stdub{\mathtt{mult}}{n-1} + 1 & \text{ if } n > 0 \\
    \end{array}
\end{equation*}

\noindent 
After solving these equations and composing the closed-form solutions,
we obtain the following closed-form functions:

\vspace*{-1mm}
\beginequationspace
\begin{equation*}
    \begin{array}{rll}
      \stdub{\mathtt{scalar}}{n,l} = & (n + 2) \times l + 1 &  \text{ if } n \geq 0 \wedge l \geq 0 \\ 
      \stdub{\mathtt{mult}}{n} = & n + 1 & \text{ if } n \geq 0 \\ \\ [-3mm]
    \end{array}
\end{equation*}
\finishequationspace

\noindent
For the parallel execution (\texttt{par}), we obtain the following
cost relations:

\vspace*{-1mm}
\beginequationspace
\begin{equation*}
    \begin{array}{rll}
      \stdub{\mathtt{scalar}}{n,l} = & 1 &  \text{ if } l = 0 \\ [-1mm]
      \stdub{\mathtt{scalar}}{n,l} = & \underline{\underline{max}}(\stdub{\mathtt{mult}}{n}, \stdub{\mathtt{scalar}}{n,l-1}) + 1 & \text{ if } l > 0 \\
      \\ [-4mm]
      \stdub{\mathtt{mult}}{n} = & 1 &  \text{ if } n = 0 \\
      \stdub{\mathtt{mult}}{n} = & \stdub{\mathtt{mult}}{n-1} + 1 & \text{ if } n > 0 \\
    \end{array}
\end{equation*}

\noindent 
After solving these equations and composing the closed forms, we
obtain the following closed-form functions:

\vspace*{-1mm}
\beginequationspace
\begin{equation*}
    \begin{array}{rll}
      \stdub{\mathtt{scalar}}{n,l} = & n + l + 1 &  \text{ if } n \geq 0 \wedge l \geq 0 \\ 
      \stdub{\mathtt{mult}}{n} = & n + 1 & \text{ if } n \geq 0 \\ \\ [-3mm]
    \end{array}
\end{equation*}
\finishequationspace

\noindent
By comparing the complexity order (in terms of resolution steps) of
the sequential execution of \texttt{scalar/3}, $O(n \cdot l)$, with
the complexity order of its parallel execution (assuming an ideal
parallel model with an unbounded number of processors) $O(n + l)$, we
can get a hint about the maximum achievable parallelization of the
program.

Another useful piece of information about \texttt{scalar/3}
that we want to infer is the maximum number of processes that may
run in parallel, considering all possible executions.
For this purpose, we define a resource in our framework named
\texttt{sthreads}.
The operation that aggregates the cost of both arguments of the
meta-predicate \texttt{\&/2}, \texttt{count\_process/3} for the
\texttt{sthreads} resource, adds the maximum number of processes for
each argument plus one additional process, corresponding to the one
created by the call to \texttt{\&/2}. The sequential cost aggregator is now
the \emph{maximum} operator, in order to keep track of the maximum
number of processes created along the different instructions of the
program executed sequentially.  Note that if the instruction
\texttt{p} executes at most $Pr_p$ processes in parallel, and the
instruction \texttt{q} executes at most $Pr_q$ processes, then the
program \texttt{p, q} will execute at most $max(Pr_p, Pr_q)$ processes
in parallel, because all the parallel processes created by \texttt{p}
will finish before the execution of \texttt{q}. Note also that for the
sequential execution of both \texttt{p} and \texttt{q},
the cost in terms of the \texttt{sthreads} resource is always zero,
because no additional process is created.

\noindent
The analysis sets up the following recurrences for the
\texttt{sthreads} resource and the predicates \texttt{scalar/3} and
\texttt{mult/3} of our example:

\vspace*{-1mm}
\beginequationspace
\begin{equation*}
    \begin{array}{rll}
      \stdub{\mathtt{scalar}}{n,l} = & 0 &  \text{ if } l = 0 \\ [-1mm]
      \stdub{\mathtt{scalar}}{n,l} = & \stdub{\mathtt{mult}}{n} + \stdub{\mathtt{scalar}}{n,l-1} + 1 & \text{ if } l > 0 \\
      \\ [-4mm]
      \stdub{\mathtt{mult}}{n} = & 0 &  \text{ if } n \geq 0 \\
    \end{array}
\end{equation*}

\noindent 
After solving these equations and composing the closed forms, we
obtain the following closed-form functions:

\vspace*{-1mm}
\beginequationspace
\begin{equation*}
    \begin{array}{rll}
      \stdub{\mathtt{scalar}}{n,l} = & l &  \text{ if } n \geq 0 \wedge l \geq 0 \\ 
      \stdub{\mathtt{mult}}{n} = & 0 & \text{ if } n \geq 0 \\ \\ [-3mm]
    \end{array}
\end{equation*}
\finishequationspace

\noindent
As we can see, this predicate will execute, in the worst case, as many
processes as there are elements in the input list.

\end{example}

\section{The Parametric Cost Relations Framework for Sequential Programs}
\label{sec:starting-cost-re-framework}
\secend

The starting point of our work is the standard
general framework described in~\cite{resource-iclp07-short} for setting up parametric relations
representing the resource usage (and size relations) of programs and predicates.
\footnote{
We give equivalent but simpler descriptions
  than in~\cite{resource-iclp07-short},
which are allowed by assuming that programs are the result of a
normalization process that makes all unifications explicit in the
clause body, so that the arguments of the clause head and the body
literals are all unique variables. We also change some notation
for readability and illustrative purposes.}

The framework is doubly parametric: first, the costs inferred are
functions of input data sizes, and second, the framework itself is
parametric with respect to the 
type of approximation made (upper or lower bounds), and to the
resource analyzed. Each concrete resource $\resr$ to be tracked is
defined by two sets of (user-provided) functions, which can be
constants, or general expressions of input data sizes: 

\vspace{-2mm}
\begin{enumerate}
\itemsep=0pt
\item \emph{Head cost 
$\headcost_{[ap,r]}(H)$:} a function that returns the amount of
resource $\resr$ used by the unification of the calling literal
(subgoal) $\predp$ and the head $H$ of a clause matching $\predp$, plus any
preparation for entering a clause (i.e., call and parameter passing
cost).

\item \emph{Predicate cost} $\Psi_{[ap,r]}(\predp, \vecx)$: it is also
possible to define the \emph{full cost} for a particular predicate
$\predp$ for resource $r$ and approximation $ap$, i.e., the function
$\Psi_{[ap,r]}(\predp):\sztypdom \rightarrow \rs$ (with the sizes of
$\predp$'s input data as parameters, $\vecx$) that returns the usage
of resource $\resr$ made by a call to this predicate. This is
specially useful for built-in or external predicates, i.e., predicates
for which the source code is not available and thus cannot be
analyzed, or for providing a more accurate function than analysis can
infer. In the implementation, this information is provided to the
analyzer through \emph{trust assertions}.

\end{enumerate}
For simplicity we only show the equations related to our standard
definition of cost. However, our framework has also been extended to allow
the inference of a more general definition of cost, called accumulated
cost, which is useful for performing static profiling, obtaining a
more detailed information regarding how the cost is distributed among
a set of user-defined \emph{cost centers}. See
\cite{staticprofiling-flops-short,gen-staticprofiling-iclp16-shortest} for more
details.

Consider a predicate $\predp$ defined by clauses $C_1, \ldots,
C_m$. Assume $\vecx$ are the sizes of $\predp$'s input
parameters. Then, the resource usage (expressed in units of resource
$r$ with approximation $ap$) of a call to $\predp$, for an input of
size $\vecx$, denoted as $C_{pred[ap,r]}(\mathtt{p,\vecx})$, can be expressed as:

\begin{equation}
  \label{eq:predcost}
  C_{pred[ap,r]}(\mathtt{p,\vecx}) =
  \mathlarger{\bigodot}_{1 \leq i \leq m}(C_{cl[ap,r]}(C_i,\vecx))
\end{equation}

\noindent
where $\mathlarger{\bigodot} = ClauseAggregator(ap,r)$ is a function that takes an
approximation identifier $ap$ and returns a function which applies
over the cost of all the clauses, $C_{cl[ap,r]}(C_i,\vecx)$, for
$1 < i < m$, in order to obtain the cost of a call to the predicate $\predp$. For
example, if $ap$ is the identifier for approximation "upper bound"
(ub), then a possible conservative definition for
$ClauseAggregator(ub,r)$ is the $\sum$ function. In this case, and
since the number of solutions generated by a predicate that will be
demanded is generally not known in advance, a conservative upper bound
on the computational cost of a predicate is obtained by assuming that
all solutions are needed, and that all clauses are executed (thus the
cost of the predicate is assumed to be the sum of the costs of all of
its clauses). However, it is straightforward to take mutual exclusion
into account to obtain a more precise estimate of the cost of a
predicate, using the maximum of the costs of mutually exclusive groups
of clauses, as done in~\cite{plai-resources-iclp14-short}.

Let us see now how to compute the resource usage of a clause. Consider
a clause $C$ of predicate $\predp$ of the form $H\neck L_1,\ldots,
L_k$ where $L_j, 1 < j < k$, is a literal (either a predicate call, or
an external or builtin predicate), and $H$ is the clause head. Assume
that $\psi_j(\vecx)$ is a tuple with the sizes of all the input
arguments to literal $L_j$, given as functions of the sizes of the
input arguments to the clause head. Note that these $\psi_j(\vecx)$
size relations have previously been computed during size analysis for
all input arguments to literals in the bodies of all clauses.  Then,
the \costrelation{}
for clause $C$ and a single call to $\predp$
(obtaining all solutions), is:
\beginequationspace
\begin{equation}
  \label{eq:clausecost}
  C_{cl[ap,r]}(C,\vecx) = \headcost_{[ap,r]}(head(C)) +\sum_{j=1}^{lim(ap,C)} sols_j(\vecx) \times C_{lit[ap,r]}(L_j,\psi_{j}(\vecx))
\end{equation}
\finishequationspace
\noindent
where $lim(ap,C)$ gives the index of the last body literal that is
called in the execution of clause $C$, $\psi_{j}(\vecx)$ are the sizes of the
input parameters of literal $L_j$, and $sols_j$ represents the product
of the number of solutions produced by the predecessor literals of $L_j$
in the clause body:
\beginequationspace
\begin{equation}
  \label{eq:solcost}
sols_j(\vecx) = \prod_{i = 1}^{j-1}s_{pred}(L_i,\psi_i(\vecx))
\end{equation}
\finishequationspace
\noindent
where $s_{pred}(L_i,\psi_i(\vecx))$ gives the number of solutions produced by $L_i$, with arguments of size $\psi_i(\vecx)$.

Finally, $C_{lit[ap,r]}(L_j, \psi_j(\vecx))$ is replaced by one of the following
expressions, depending on $L_j$:
\vspace*{-3mm}
\begin{itemize}
\item If $L_j$ is a call to a predicate $q$ which is in the same
strongly connected component as $\predp$ (the predicate under analysis),
then $C_{lit[ap,r]}(L_j, \psi_j(\vecx))$ is replaced by the symbolic
call $C_{pred[ap,r]}(\mathtt{q,\psi_j(\vecx)})$, giving rise to a
recurrence relation that needs to be bounded with a closed-form
expression by the solver afterwards.
\item If $L_j$ is a call to a predicate $q$ which is in a different
strongly connected component than $\predp$, then $C_{lit[ap,r]}(L_j,
\psi_j(\vecx))$ is replaced by the closed-form expression that bounds
$C_{pred[ap,r]}(\mathtt{q,\psi_j(\vecx)})$. The analysis guarantees
that this expression has been inferred beforehand, due to the fact
that the analysis is performed for each strongly connected component,
in a reverse topological order.
\item If $L_j$ is a call to a predicate q, whose cost is specified
(with a trust assertion) as $\Psi_{[ap,r]}(q,\vecx)$, then
$C_{lit[ap,r]}(L_j, \psi_j(\vecx))$ is replaced by the expression
$\Psi_{[ap,r]}(q,\psi_j(\vecx))$.
\end{itemize}

\secbeg
\section{Our Extended Resource Analysis Framework for Parallel Programs}
\secend
\label{sec:extended-framework-parallel}
In this section, we
describe how we extend the resource
analysis framework detailed above, in order to handle logic programs
with Independent And-Parallelism, using the binary parallel
\texttt{\&/2} operator.  First, we introduce a new general parameter
that indicates the execution model the analysis has to consider.
For our current prototype, we have defined two different execution
models: standard \emph{sequential} execution, represented by
\texttt{$seq$}, and an abstract parallel execution model, represented
by \texttt{$par(n)$}, where $n \in \mathbb{N} \cup \{\infty\}$.  The
abstract execution model \texttt{$par(\infty)$} is similar to the
\emph{work} and \emph{depth} model, presented
in~\cite{blelloch96nesl-short} and used
extensively in previous work such
as~\cite{hoffman2015-short}. Basically, this model is based on
considering an unbounded number of available processors to infer
bounds on the depth of the computation tree. The \emph{work} measure
is the amount of work to be performed considering a sequential
execution.  These two measures together give an idea on the impact of
the parallelization of a particular program.
The abstract execution model \texttt{$par(n)$}, where $n \in
\mathbb{N}$, assumes a finite number $n$ of processors.

In order to obtain the cost of a predicate, equation~(\ref{eq:predcost})
remains almost identical, with the only difference of adding the new
parameter to indicate the execution model.

Now we address how to set the cost for clauses. In this case, 
equation~(\ref{eq:clausecost}) is extended with the execution model $ex$, and
also the $\Sigma$ default sequential cost aggregation is replaced by a
parametric associative operator $\bigoplus$, that depends on the
resource being defined, the approximation, and the execution model.
For $ex \equiv par(\infty)$ or $ex \equiv seq$,
the following equation is set up:

\vspace*{-3mm}
\beginequationspace
\begin{equation}
  C_{cl[ap,r,ex]}(C,\vecx) = \headcost_{[ap,r]}(head(C)) +\mathlarger{{\bigoplus_{j=1}^{lim(ap,ex,C)}}} (sols_j(\vecx) \times C_{lit[ap,r,ex]}(L_j,\psi_j(\vecx)))
\end{equation}
\finishequationspace

Note that the cost aggregator operators must depend on the resource
$r$ (besides the other parameters). For example, if $r$ is
\emph{execution time}, then the cost of executing two tasks in
parallel must be aggregated by taking the maximun of the execution
times of the two tasks. In contrast, if $r$ is \emph{energy
  consumption}, then the aggregation is the addition of the energy of
the two tasks.

Finally, we extend how the cost of a literal $L_i$, expressed as
$C_{lit[ap,r,ex]}(L_i,\psi_i(\vecx))$, is set up. The previous
definition is extended considering the new case where the literal is a
call to the \emph{meta-predicate} \texttt{\&/2}. In this case, we add
a new parallel aggregation associative operator, denoted by
$\bigotimes$. Concretely, if $L_i = B_1 \& B_2 $, where $B_1$ and
$B_2$ are two sequences of goals, then:

\begin{equation}
C_{lit[ap,r,ex]}(B_1 \& B_2,\vecx) = C_{body[ap,r,ex]}(B_1,\vecx) \bigotimes C_{body[ap,r,ex]}(B_2,\vecx)
\end{equation}

\begin{equation}
  C_{body[ap,r,ex]}(B,\vecx) = \mathlarger{{\bigoplus_{j=1}^{lim(ap,ex,B)}}} (sols_j(\vecx) \times C_{lit[ap,r,ex]}(L^B_j,\psi_j(\vecx)))
\end{equation}

\noindent
where $B = L^B_1,\cdots,L^B_m$.

Consider now the execution model $ex \equiv par(n)$, where $n \in
\mathbb{N}$ (i.e., assuming a finite number $n$ of processors), and a
recursive parallel predicate $\predp$ that creates a parallel task
$\predq_i$ in each recursion $i$. Assume that we are interested in
obtaining an upper bound on the cost of a call to $\predp$, for an
input of size $\vecx$.  We first infer the number $k$ of parallel
tasks created by $\predp$ as a function of $\vecx$.  This can be
easily done by using our cost analysis framework and providing the
suitable assertions for inferring a resource named ``$ptasks$.''
Intuitivelly, the ``counter'' associated to such resource must be
incremented by the (symbolic) execution of the \texttt{\&/2} parallel
operator. More formally, $k = C_{pred[ub,ptasks]}(\predp,\vecx)$.  To
this point, an upper bound $m$ on the number of tasks executed by any
of the $n$ processors is given by $m = \lceil \frac{k}{n} \rceil$.
Then, an upper bound on the cost (in terms of resolution steps, i.e.,
$r=steps$) of a call to $\predp$, for an input of size $\vecx$ can be
given by:

\begin{equation}
 \label{eq:predcost:finite:processors}
  C_{pred[ub,r,par(n)]}(\predp,\vecx) = C^{u} + Spaw^{u}
\end{equation}

\noindent
where $C^{u}$ can be computed in two possible ways: $C^{u} =
\sum_{i=1}^{m} C_{i}^{u}$; or $C^{u} = m \ C_1^{u}$, where $C_i^{u}$
denotes an upper bound on the cost of parallel task $\predq_i$, and
$C_{1}^{u}, \ldots, C_{k}^{u}$ are ordered in descending order of
cost. Each $C_i^{u}$ can be considered as the sum of two components:
$C_i^{u} = Sched_i^{u} + T_i^{u}$, where $Sched_i^{u}$ denotes the
cost from the point in which the parallel subtask $\predq_i$ is
created until its execution is started by a processor (possibly the
same processor that created the subtask), i.e.\ the cost of task
preparation, scheduling, communication overheads, etc.
$T_i^{u}$ denotes the
cost of the execution of $\predq_i$ disregarding all the overheads
mentioned before, i.e., $T_i^{u} =
C_{pred[ub,r,seq]}(\predq,\psi_q(\vecx))$, where $\psi_q(\vecx)$ is a
tuple with the sizes of all the input arguments to predicate $\predq$
in the body of $\predp$.  $Spaw^{u}$ denotes an upper bound on the
cost of creating the $k$ parallel tasks $\predq_i$.  It will be
dependent on the particular system in which $\predp$ is going to be
executed. It can be a constant, or a function of several parameters,
(such as input data size, number of input arguments, or number of
tasks) and can be experimentally determined.

In addition, we propose a method for finding closed-form functions for
cost relations that arise in the analysis of parallel programs, where 
 the $max$ function usually plays a role both as parallel and sequential cost aggregation operation, i.e, as $\bigotimes$ and $\bigoplus$ respectively. In the following subsection, we detail these methods.

\subsection{Solving Cost Recurrence Relations Involving $max$ Operation}

Automatically finding closed-form upper and lower bounds for
recurrence relations is an uncomputable problem.
For some special classes of recurrences, exact solutions are known, for
example for linear recurrences with one variable. For some other
classes, it is possible to apply transformations to fit a class of
recurrences with known solutions, even if this transformation obtains
an appropriate approximation rather than an equivalent expression.

Particularly for the case of analyzing independent and-parallel logic
programs, nonlinear recurrences involving the $max$ operator are quite
common. For example, if we are analyzing elapsed time of a
parallel logic program, a proper parallel aggregation operator is the
maximum between the times elapsed for each literal running in
parallel. To the best of our knowledge, no general solution exists for
recurrences of this particular type. However, in this paper we
identify some common cases of this type of recurrences, for which we
obtain closed forms that are proven to be correct. In this section, we
present these different classes, together with the corresponding
method to obtain a correct bound.

Consider the following function $f: \mathbb{N}^m \to
\mathbb{R}^{+}$, defined as a general form of a first-order
recurrence equation with a $max$ operator:
\begin{equation}
\label{defmaxrec}
f(\vecx) = \begin{cases} 
      max(C,f(\vecx_{|i} - 1)) + D & x_i > \Theta  \\
      B & x_i \leq \Theta 
    \end{cases}
 \end{equation}
\noindent
where $\Theta \in \mathbb{N}$. $C, D$ and $B$ are arbitrary expressions possibly depending on $\vecx$. Note that $\vecx = x_1,x_2,\cdots,x_m $. We define $\vecx_{|i} - 1 = x_1,\cdots,x_i - 1, \cdots, x_m$, for a given $i$, $1 \leq i \leq m$. 
If $C$ and $D$ do not depend on $x_i$, then $C$ and $D$ do not change
through the different recursive instances of $f$. In this case, a
closed-form upper bound is defined by the following theorem (whose
proof is included in~\ref{appendix:proof-one}):

\begin{theorem}
    \label{t1}
  Given $f: \mathbb{N}^m \to \mathbb{N}$ as defined in
  (\ref{defmaxrec}), where $C$ and $D$ are non-decreasing functions of
  $\vecx \setminus x_i$. Then, $\forall \vecx$: 
          \[ f(\vecx) = f^{'}(\vecx) = \begin{cases} 
            max(C,B) + (x_i - \Theta) \cdot D & x_i > \Theta \\
            B  & x_i \leq \Theta  \\
            \end{cases}
          \]
\end{theorem}

For the case where $C = g(\vecx)$ and $D = h(\vecx)$ are functions
non-decreasing on $x_i$, then the upper bound is given by the
following closed form:

\begin{theorem}
\label{t2}
  Given $f: \mathbb{N}^m \to \mathbb{N}$ as defined in
(\ref{defmaxrec}), where $g$ and $h$ are functions of $\vecx$,
non-decreasing on $x_i$. Then, $\forall \vecx$: 
\[ f(\vecx) \leq  f^\prime(\vecx) = 
  \begin{cases}
    max(g(\vecx),B) + (x_i - \Theta - 1) \times max(g(\vecx),h(\vecx_{|i}-1)) + h(\vecx_{|i}) & x_i > \Theta \\
    B & x_i \leq \Theta
  \end{cases}
\]
\end{theorem}

\noindent
The proof of this Theorem is included in~\ref{appendix:proof-two}.

For the remaining cases, where a $max(e_1, e_2)$ appears, we try to
eliminate the $max$ operator by proving either $e_1 \leq e_2$ or $e_2
\leq e_1$, for any input. In order to do that, we use the function
comparison capabilities of CiaoPP, presented in
\cite{resource-verif-iclp2010-short,resource-verification-tplp18-shortest}. In cases
where $e_1$ and/or $e_2$ contains non-closed recurrence functions, we
use the Z3 SMT solver \cite{z3-shorter} to find, if possible, a proof either for $e_1 \leq
e_2$ or $e_2 \leq e_1$, treating the non-closed functions as
uninterpreted functions, assuming that they are positive and
non-decreasing. As the algorithm used by SMT solvers in this case is
not guaranteed to terminate, we set a timeout. In the worst case,
when no proof is found, then we replace the $max$ operator with an
addition, loosing precision but still finding safe upper bounds.

\secbeg
\section{Implementation and Experimental Results}
\secend
\label{sec:implem-experiments}

We have implemented a prototype of our approach, leveraging the
existing resource usage analysis framework of \ciaopp. The
implementation basically consists of the parameterization of the operators used
for sequential and parallel cost aggregation, i.e., for the aggregation
of the costs corresponding to the arguments of \texttt{,/2} and
\texttt{\&/2}, respectively. This allows the user to define resources
in a general way, taking into account the underlying execution model.

\begin{table*}[t]
\begin{tabular}{r|l}
  \kbd{map\_add1/2}   & Parallel increment by one of each element of a list. \\
  \kbd{fib/2}         & Parallel computation of the nth Fibonacci number. \\
  \kbd{mmatrix/3}     & Parallel matrix multiplication. \\
  \kbd{blur/2}        & Generic parallel image filter. \\
  \kbd{add\_mat/3}    & Matrix addition. \\
  \kbd{intersect/3}   & Set intersection. \\
  \kbd{union/3}       & Set union. \\
  \kbd{diff/3}        & Set difference. \\
  \kbd{dyade/3}       & Dyadic product of two vectors. \\
  \kbd{dyade\_map/3}  & Dyadic product applied on a set of vectors. \\
  \kbd{append\_all/3} & Appends a prefix to each list of a list of lists. \\
\end{tabular}
\caption{Description of the benchmarks.}
\label{fig:descbenchmarks}
\end{table*}

\begin{table*}[t]
  \renewcommand{\times}{\cdot}
  \renewcommand{\cfrac}{\tfrac}
  \scriptsize
    \caption{Resource usage inferred for Independent And-Parallel Programs.}
  \label{tbl:cost-expressions-exact}
  \setlength\tabcolsep{1pt}
  \renewcommand{\arraystretch}{1.1}
  \begin{tabular}{|p{3cm}|r|p{5.5cm}|c|c|}
    \cline{1-5}
\benchmarks & \tabres  & \centering{\colexp} & $\bigo$ &  \coltimes  \\
    \cline{1-5}
    \multirow{4}{*}{\texttt{map\_add1(x)}}
            & \textsf{\seqcost} & $2 \cdot \mlength{x} + 1$ & $\mathcal{O}(\mlength{x})$  & \multirow{4}{*}{35.57} \\
            & \textsf{\parcost} & $2 \cdot \mlength{x} + 1$ & $\mathcal{O}(\mlength{x})$  &  \\
            & \textsf{\parthreads} & $ \mlength{x} $ & $\mathcal{O}(\mlength{x})$  & \\
        \cline{1-5}
    \multirow{4}{*}{\texttt{fib(x)}}
            & \textsf{\seqcost} & $F(\mint{x})+L(\mint{x}) - 1$ &  $\mathcal{O}(2^{\mint{x}})$  & \multirow{4}{*}{52.66} \\
            & \textsf{\parcost} & $\mint{x} + 1$ &  $\mathcal{O}(\mint{x})$ &  \\
            & \textsf{\parthreads} & $F(\mint{x}) + L(\mint{x}) - 1$ &  $\mathcal{O}(2^{\mint{x}})$ &  \\
        \cline{1-5}
  \multirow{4}{*}{\texttt{mmatrix($m_1,n_1,m_2,n_2$)}}
            & \textsf{\seqcost} & $ \mint{n_2} \cdot \mint{m_2} \cdot \mint{m_1}+2 \cdot \mint{m_2} \cdot \mint{m_1} + 2 \cdot \mint{m_1}+1 $ & $\mathcal{O}(\mint{n_2} \cdot \mint{m_2} \cdot \mint{m_1})$ & \multirow{4}{*}{220.9} \\
            & \textsf{\parcost} & $ 2 \cdot \mint{m_1} + \mint{n_1} + 1 $ & $\mathcal{O}(\mint{n_1} + \mint{m_1})$ & \\
            & \textsf{\parthreads} & $ \mint{m_2} \cdot \mint{m_1} + \mint{m_1} $ & $\mathcal{O}(\mint{m_2} \cdot \mint{m_1})$ & \\
        \cline{1-5}
    \multirow{4}{*}{\texttt{blur(m,n)}}
            & \textsf{\seqcost} & $2\cdot\mint{m} \cdot \mint{n}+2 \cdot \mint{n} + 1$ & $\mathcal{O}(\mint{m} \cdot \mint{n})$ & \multirow{4}{*}{123.321} \\
            & \textsf{\parcost} & $ 2 \cdot \mint{m} +  2 \cdot \mint{n} + 1 $ & $\mathcal{O}(\mint{m} +  \mint{n})$ &  \\
            & \textsf{\parthreads} & $ \mint{n} $ & $\mathcal{O}(\mint{n})$ & \\
        \cline{1-5}
    \multirow{4}{*}{\texttt{add\_mat(m,n)}}
            & \textsf{\seqcost} & $\mint{m} \cdot \mint{n} + 2 \cdot \mint{n} + 1$ & $\mathcal{O}(\mint{m} \cdot \mint{n})$ & \multirow{4}{*}{62.72} \\
            & \textsf{\parcost} & $ \mint{m} +  2 \cdot \mint{n} + 1 $ & $\mathcal{O}(\mint{m} + \mint{n})$ &  \\
            & \textsf{\parthreads} & $ \mint{n} $ & $\mathcal{O}(\mint{n})$ & \\
        \cline{1-5}
    \multirow{4}{*}{\texttt{intersect($a,b$)}}
            & \textsf{\seqcost} & $ \mlength{a} \cdot \mlength{b} + 3 \cdot \mlength{a} + 3$ & $\mathcal{O}(\mlength{x})$ & \multirow{4}{*}{191.16} \\
            & \textsf{\parcost} & $ \mlength{b} + 2 \cdot \mlength{a} + 3$ & $\mathcal{O}(\mint{n})$ &  \\
            & \textsf{\parthreads} & $ \mlength{a} $ & $\mathcal{O}(\mlength{x})$ & \\
    \cline{1-5}
    \multirow{4}{*}{\texttt{union($a,b$)}}
            & \textsf{\seqcost} & $ \mlength{a} \cdot \mlength{b} + 3 \cdot \mlength{a} + 3$ & $\mathcal{O}(\mlength{a} \cdot \mlength{b})$ & \multirow{4}{*}{193.37} \\
            & \textsf{\parcost} & $ 2 \cdot \mlength{b} + 2 \cdot \mlength{a} + 3$ & $\mathcal{O}(\mlength{a} + \mlength{b})$ &  \\
            & \textsf{\parthreads} & $ \mlength{a} $ & $\mathcal{O}(\mlength{a})$ & \\
    \cline{1-5}
    \multirow{4}{*}{\texttt{diff($a,b$)}}
            & \textsf{\seqcost} & $ \mlength{a} \cdot \mlength{b} + 3 \cdot \mlength{a} + 3$ & $\mathcal{O}(\mlength{a} \cdot \mlength{b})$ & \multirow{4}{*}{191.16}  \\
            & \textsf{\parcost} & $ \mlength{b} + 2 \cdot \mlength{a} + 3$ & $\mathcal{O}(\mlength{a} + \mlength{b})$ &  \\
            & \textsf{\parthreads} & $ \mlength{a} $ & $\mathcal{O}(\mlength{a})$ & \\
    \cline{1-5}
    \multirow{4}{*}{\texttt{dyade($a,b$)}}
            & \textsf{\seqcost} & $ \mlength{a} \cdot \mlength{b} + 2 \cdot \mlength{a} + 1$ & $\mathcal{O}(\mlength{a} \cdot \mlength{b})$ & \multirow{4}{*}{71.08} \\
            & \textsf{\parcost} & $ \mlength{b} + \mlength{a} + 1$ & $\mathcal{O}(\mlength{a} + \mlength{b})$ &  \\
            & \textsf{\parthreads} & $ \mlength{a} $ & $\mathcal{O}(\mlength{a})$ & \\
    \cline{1-5}
    \multirow{4}{*}{\texttt{dyade\_map($l,m$)}}
            & \textsf{\seqcost} & $ \mint{max(m)} \cdot \mlength{m} \cdot \mlength{l}+2 \cdot \mlength{m} \cdot \mlength{l}+2 \cdot \mlength{m}+1 $ & $\mathcal{O}(\mint{max(m)} \cdot \mlength{m} \cdot \mlength{l})$ & \multirow{4}{*}{248.39} \\
            & \textsf{\parcost} & $ \mint{max(m)} + \mlength{m} + \mlength{l} + 1  $ & $\mathcal{O}(\mint{max(m)} + \mlength{m} + \mlength{l})$ & \\
            & \textsf{\parthreads} & $ \mlength{l} \cdot \mlength{m} + \mlength{l}  $ & $\mathcal{O}(\mlength{m} \cdot \mlength{l})$ & \\
    \cline{1-5}
    \multirow{4}{*}{\texttt{append\_all($l,m$)}}
            & \textsf{\seqcost} & $ \mlength{l} \cdot \mlength{m} + 2 \cdot \mlength{m} + 1 $ & $\mathcal{O}(\mlength{l} \cdot \mlength{m})$ & \multirow{4}{*}{108.4} \\
            & \textsf{\parcost} & $ \mlength{l} + \mlength{m} + 1 $ & $\mathcal{O}(\mlength{l} + \mlength{m})$ &  \\
            & \textsf{\parthreads} & $ \mlength{m} $ & $\mathcal{O}(\mlength{m})$ & \\
    \cline{1-5}
\end{tabular}
\label{table:ExpResults}
\begin{minipage}{\textwidth}  
\scriptsize 
\begin{itemize}
\item $F(n)$ represents the nth element of the Fibonacci sequence.
\item $L(n)$ represents the nth Lucas number.
\item $l_n, i_n$ represent the size of $n$ in terms of the metrics
\emph{length} and \emph{int}, respectively.
\end{itemize}
\end{minipage}
\end{table*}

\begin{table*}[t]
  \renewcommand{\times}{\cdot}
  \renewcommand{\cfrac}{\tfrac}
  \scriptsize
  \caption{Resource usage inferred for a bounded number of processors.}
  \label{tbl:cost-expressions-numberprocessors}
  \setlength\tabcolsep{1pt}
  \renewcommand{\arraystretch}{1.1}
  \begin{tabular}{|p{3cm}|p{5.5cm}|c|c|}
    \cline{1-4}
\benchmarks & \centering{\colexp} & $\bigo$ &  \coltimes  \\
    \cline{1-4}
    \texttt{map\_add1(x)}
            & $  2 \cdot \lceil \frac{\mlength{x}}{p} \rceil + 1 $ & $\mathcal{O}(\lceil \frac{\mlength{x}}{p} \rceil)$ & 54.36 \\
    \texttt{blur(m,n)}
              & $ 2 \cdot \lceil \frac{\mint{n}}{p} \rceil \cdot \mint{m} +2 \cdot \lceil \frac{\mint{n}}{p} \rceil  + 1 $ & $\mathcal{O}(\lceil \frac{\mint{n}}{p}  \rceil \cdot \mint{m})$  & 205.97 \\
    \texttt{add\_mat(m,n)}
            & $ \lceil \frac{\mint{n}}{p} \rceil \cdot \mint{m} + 2 \cdot \lceil \frac{\mint{n}}{p} \rceil + 1 $ &   $\mathcal{O}(\lceil \frac{\mint{n}}{p}  \rceil \cdot \mint{m})$  & 185.89  \\
     \texttt{intersect($a,b$)}
              & $ \lceil \frac{\mlength{a}}{p} \rceil \cdot \mlength{b} + 2 \cdot \lceil \frac{\mlength{a}}{p} \rceil + \mlength{a} + 2 $ & $\mathcal{O}(\lceil \frac{\mlength{a}}{p}  \rceil \cdot \mlength{b})$  & 330.47 \\
     \texttt{union($a,b$)}
              & $ \lceil \frac{\mlength{a}}{p} \rceil \cdot \mlength{b} + 2 \cdot \lceil \frac{\mlength{a}}{p} \rceil + \mlength{a} + \mlength{b} + 2 $ & $\mathcal{O}(\lceil \frac{\mlength{a}}{p}  \rceil \cdot \mlength{b})$ & 311.3 \\
     \texttt{diff($a,b$)}
              & $ \lceil \frac{\mlength{a}}{p} \rceil \cdot \mlength{b} + 2 \cdot \lceil \frac{\mlength{a}}{p} \rceil + \mlength{a} + 2 $ & $\mathcal{O}(\lceil \frac{\mlength{a}}{p}  \rceil \cdot \mlength{b})$ & 339.01 \\
    \texttt{dyade($a,b$)}
              & $ \lceil \frac{\mlength{a}}{p} \rceil \cdot \mlength{b} + 2 \cdot \lceil \frac{\mlength{a}}{p} \rceil + 1$ & $\mathcal{O}(\lceil \frac{\mlength{a}}{p}  \rceil \cdot \mlength{b})$ & 120.93 \\
     \texttt{append\_all($l,m$)}
              & $ \lceil \frac{\mlength{m}}{p} \rceil \cdot \mlength{l} + 2 \cdot \lceil \frac{\mlength{m}}{p} \rceil + 1 $ & $\mathcal{O}(\lceil \frac{\mlength{m}}{p}  \rceil \cdot \mlength{l})$ & 117.8 \\
    \cline{1-4}
\end{tabular}
\begin{minipage}{\textwidth}  
\scriptsize 
$\predp$ is defined as the minimum between the number of processors and $\parthreads$.
\end{minipage}
\end{table*}

We selected a set of benchmarks that exhibit different common
parallel patterns, briefly described in
Table~\ref{fig:descbenchmarks}, together with the definition of a set
of resources that help understand the overall behavior of the
parallelization.\footnote{We will be able to extend the experiments to a bigger
  set of benchmarks for the talk at the conference and the
  post-proceedings submission if the paper is accepted.}
Table~\ref{table:ExpResults} shows some results of the experiments that
we have performed with our prototype implementation.  Column
\benchmarks{} shows the main predicates analyzed for each
benchmark. Set operations (\texttt{intersect}, \texttt{union} and
\texttt{diff}), as well as the programs \texttt{append\_all},
\texttt{dyade} and \texttt{add\_mat}, are Prolog versions of the
benchmarks analyzed in~\cite{hoffman2015-short}, which is the closest
related work we are aware of.

Column~\tabres{} indicates the name of each of the resources inferred
for each benchmark: \emph{sequential resolution steps} (\seqcost),
\emph{parallel resolution steps} assuming an unbounded number of
processors (\parcost), and \emph{maximum number of processes executing
  in parallel} (\parthreads). The latter gives an indication of the
maximum parallelism that can potentially be exploited.
Column~\colexp{} shows the upper bounds obtained for each of the
resources indicated in Column~\tabres.  While in the experiments both
upper and lower bounds were inferred, for the sake of brevity, we only
show upper bound functions.  Column~\bigo{} shows the complexity
order, in big O notation, corresponding to each resource.  For all the
benchmarks in Table~\ref{table:ExpResults}
we obtain the exact complexity orders.
We also obtain the same complexity order as
in~\cite{hoffman2015-short} for the Prolog versions of the benchmarks taken
from that work.
Finally, Column~\atime{} shows the analysis times in milliseconds,
which are quite reasonable.
The results show that most of the benchmarks have different
asymptotic behavior in the 
sequential and parallel execution models. In particular, for
\texttt{fib(x)}, the analysis infers an exponential upper bound for
sequential execution steps, and a linear upper bound for parallel
execution steps. As mentioned before, this is an upper bound for an
ideal case, assuming an unbounded number of
processors. Nevertheless, such upper-bound information is useful for
understanding how the cost behavior evolves in architectures with
different levels of parallelism. In addition, this \emph{dual} cost
measure can be combined together with a bound on the number of
processors in order to obtain a general asymptotic upper bound (see
for example Brent's Theorem~\cite{harper2016}, which is also mentioned
in~\cite{hoffman2015-short}). The program \texttt{map\_add1(l)} exhibits a
different behavior: both sequential and parallel upper bounds are
linear. This happens because we are considering \emph{resolution
  steps}, i.e., we are counting each head unification produced from an
initial call \texttt{map\_add1(l)}. Even under the parallel execution
model, we have a chain of head unifications whose length depends
linearly on the length of the input list. It follows from the results
of this particular case that the parallelization will not be useful
for improving the number of resolution steps performed in parallel.

Another useful information inferred in our experiments is the maximum
number of processes that can be executed in parallel, represented by
the resource named \parthreads. We can see that for most of our
examples the analysis obtains a linear upper bound for this resource,
in terms of the size of some of the inputs. For example, the execution
of \texttt{intersect(a,b)} (parallel set intersection) will create
\emph{at most} $\mlength{a}$ processes, where $\mlength{a}$ represents
the length of the list $a$. For other examples, the analysis shows a
quadratic upper bound (as in \texttt{mmatrix}), or even exponential
bounds (as in \texttt{fib}).
The information about upper bounds on the maximum level of parallelism
required by a program is useful for understanding its scalability in
different parallel architectures, or for optimizing the number of
processors that a particular call will use, depending on the size of
the input data.

Finally, the results of our experiments considering a bounded number
of processors are shown in
Table~\ref{tbl:cost-expressions-numberprocessors}.

\secbeg
\section{Related Work}
\secend
\label{sec:related-work}

Our approach is an extension of an existing cost analysis framework
 for sequential logic programs \cite{low-bounds-ilps97-short,staticprofiling-flops-short,resource-verification-tplp18-shortest}, 
 which extends the classical cost analysis 
 techniques based on
 setting up and solving
recurrence relations, pioneered by~\cite{Wegbreit75-short-plus}, with 
solutions for relations involving ~\texttt{max} and
\texttt{min} functions.  The framework handles characteristics such
as backtracking, multiple solutions (i.e., non-determinism),  
failure, and inference of both upper and lower bounds including non-polynomial bounds.
These features are inherited by our approach, and are absent from 
other approaches to parallel cost analysis in the literature.

The most closely-related work to our approach is~\cite{hoffman2015-short}, 
which describes  an automatic analysis for
deriving bounds on the worst-case evaluation cost of first order
functional programs. The analysis derives bounds under an abstract \emph{dual}
cost model based on two measures: \emph{work} and \emph{depth}, which
over-approximate the sequential and parallel evaluation cost of
programs, respectively, considering an unlimited number of
processors. Such an abstract cost model was introduced
by~\cite{blelloch96nesl-short} to formally analyze parallel programs.
The work is based on type judgments 
annotated with a cost metric, which generate a set of inequalities
which are then solved by linear programming techniques.  Their analysis is
only able to infer multivariate resource polynomial bounds, while
non-polynomial bounds are left as future
work. 
In~\cite{Hoefler:2014:ACA:2612669.2612685-short} the authors propose an
automatic analysis based on the \emph{work} and \emph{depth} model,
for a simple imperative language with explicit parallel loops.

There are other approaches to cost analysis of parallel and distributed 
systems, based on different models of computation than the independent
and-parallel model in our work. 
In~\cite{AlbertAGZ11-short} the authors present a static analysis which is able to
infer upper bounds on the maximum number of \emph{active} (i.e, not
finished nor suspended) processes running in parallel, and the total
number of processes created for imperative \emph{async-finish}
parallel programs. 
The approach described in~\cite{costabs-11-semi-short} uses recurrence (cost)
relations to derive upper bounds on the cost of concurrent
object-oriented programs, with shared-memory communication and future
variables. They address concurrent execution for loops with a
semi-controlled scheduling, i.e., with no arbitrary
interleavings. In~\cite{Albert:2018:PCA:3293495.3274278-short} the authors address
the cost of parallel execution of object-oriented distributed programs. The approach 
is to identify the synchronization points in the program, use serial cost analysis
of the blocks between these points, and then, exploiting techniques just mentioned, construct a graph structure
to capture the possible parallel execution of the program. The path of
maximal cost is then computed.  The allocation of tasks to processors (called ``locations") is part of
the program in these works, and so although independent and-parallel programs could be modelled 
in this computation style, it is not
directly comparable to our more abstract model of parallelism.

Solving, or safely bounding recurrence relations with~\texttt{max} and
\texttt{min} functions has been addressed mainly for recurrences
derived from divide-and-conquer
algorithms~\cite{Alonso:1995:MDM:213240.213251-short,WANG2000377-short,HWANG20031475-short}.
In our experience, our method is able to obtain more accurate bounds.

\secbeg
\section{Conclusions}
\secend
\label{sec:conclusions}
We have presented a novel, general, and flexible analysis framework
that can be instantiated for estimating the resource usage of parallel
logic programs, for a wide range of resources, platforms, and
execution models. To the best of our knowledge, this is the first
approach to the cost analysis of \emph{parallel logic programs}. Such
estimations include both lower and upper bounds, given as functions on
input data sizes. In addition, our analysis also infers other
information which is useful for improving the exploitation and
assessing the potential and actual parallelism of a program. We have
also developed a method for solving the cost relations that arise in
this particular type of analysis, which involve the $max$
function. Finally, we have developed a prototype implementation of our
general framework, instantiated it for the analysis of logic programs
with Independent And-Parallelism, and performed an experimental
evaluation, obtaining very encouraging results w.r.t. accuracy and
efficiency.

\vspace*{4mm}
\noindent
\textbf{Acknowledgements:} Research partially funded by EU FP7
agreement no 318337 \emph{ENTRA}, Spanish MINECO TIN2015-67522-C3-1-R
\emph{TRACES} project, the Madrid M141047003 \emph{N-GREENS} program
and \emph{BLOQUES-CM} project, and the TEZOS Foundation \emph{TEZOS}
project.

\begin{small}

\end{small}

\appendix

\clearpage 
  \vspace*{-20mm}

  \centerline{{\LARGE \textbf{Appendices}}}

  \ \\ [-3mm]

\section{Proof for Theorem \ref{t1}}
\label{appendix:proof-one}

\noindent
\begin{ntheorem}
  Given $f: \mathbb{N}^m \to \mathbb{N}$ as defined in
  (\ref{defmaxrec}), where $C$ and $D$ are non-decreasing functions of
  $\vecx \setminus x_i$. Then, $\forall \vecx$: 
          \[ f(\vecx) = f^{'}(\vecx) = \begin{cases} 
            max(C,B) + (x_i - \Theta + 1) \cdot D & x_i > \Theta \\
            B  & x_i \leq \Theta  \\
            \end{cases}
          \]
\end{ntheorem}
\begin{proof}
The proof for the case $x_i \leq \Theta$ is trivial.

\noindent
In the following, we prove the theorem for $x_i > \Theta$, or equivalently,
for $x_i \geq \Theta + 1$. The proof is by induction on this subset. \\

\noindent
\textbf{Base Case.} We have to prove that $f(x_1,\cdots,x_{i-1},\Theta+1,\cdots,x_m) = f^\prime(x_1,\cdots,x_{i-1},\Theta+1,\cdots,x_m)$. Using the definition of $f$ and $f^\prime$ we have
that

\begin{equation*}
\begin{split}
  f(x_1,\cdots,x_{i-1},\Theta+1,\cdots,x_m) & = max(C,f(x_1,\cdots,x_{i-1},\Theta,\cdots,x_m)) + D \\
  & = max(C,B) + D \\
  f^{\prime}(x_1,\cdots,x_{i-1},\Theta+1,\cdots,x_m) & = max(C,B) + (\Theta + 1 - \Theta) \cdot D \\
  & = max(C,B) + D
\end{split}
\end{equation*}

\noindent

\noindent
\textbf{General Case.} Assuming \\ $f(x_1,\cdots,x_{i-1},x_i,\cdots,x_m) = f^{\prime}(x_1,\cdots,x_{i-1},x_i,\cdots,x_m)$, we need to
prove that $f(x_1,\cdots,x_{i-1},x_i+1,\cdots,x_m) = f^{\prime}(x_1,\cdots,x_{i-1},x_i+1,\cdots,x_m)$.
By induction hypothesis we have that:

\begin{equation*}
\begin{split}
  f(x_1,\cdots,x_{i-1},x_i + 1,\cdots,x_m)  & = max(C,f(x_1,\cdots,x_{i-1},x_i,\cdots,x_m)) + D \\ 
  & = max(C,max(C,B) + (x_i - \Theta) \cdot D) + D \\
  & = max(C,B) + (x_i - \Theta) \cdot D + D \\
  & = max(C,B) + (x_i  - \Theta + 1) \cdot D \\
  & = f^{\prime}(x_1,\cdots,x_{i-1},x_i + 1,\cdots,x_m)
\end{split}
\end{equation*}

\end{proof}

\clearpage
\section{Proof of Theorem \ref{t2}}
\label{appendix:proof-two}
For all $a, b, c \in \mathbb{N}\cup \{0\}$, the following properties hold:
\begin{itemize}
    \item Commutative: $max(a,b) = max(b,a)$
    \item Associative: $max(a,max(b,c)) = max(max(a,b),c)$
    \item Idempotent: $max(a,a) = a$
\end{itemize}

\begin{lemma}
\label{lemmamax}
$\forall a, b, c \in \mathbb{N}: max(a,b + c) \leq max(a,b) + max(a,c)$
\end{lemma}
\begin{lemma}
\label{lemmamax2}
$\forall a, b, c, d \in \mathbb{N}: a \leq c \wedge b \leq d \implies max(a,b) \leq max(c,d)$
\end{lemma}

\noindent

\begin{ntheorem}
  Given $f: \mathbb{N}^m \to \mathbb{N}$ as defined in
(\ref{defmaxrec}), where $g$ and $h$ are functions of $\vecx$,
non-decreasing on $x_i$. Then, $\forall \vecx$: 
\[ f(\vecx) \leq  f^\prime(\vecx) = 
  \begin{cases}
    max(g(\vecx),B) + (x_i - \Theta - 1) \times max(g(\vecx),h(\vecx_{|i}-1)) + h(\vecx_{|i}) & x_i > \Theta \\
    B & x_i \leq \Theta
  \end{cases}
\]
\end{ntheorem}

\begin{proof}
The proof for the case $x_i \leq \Theta$ is trivial.

\noindent
In the following, we prove the theorem for $x_i > \Theta$, or equivalently,
for $x_i \geq \Theta + 1$. The proof is by induction on this subset.
For brevity, we only show the argument corresponding to the position
of $x_i$ in $\vecx$. However, the proof is still valid considering all
of the arguments.
\\

\noindent

\textbf{Base Case.} We have to prove that $f(\Theta+1) \leq f^\prime(\Theta+1)$. Using the definition of $f$ and $f^\prime$ we have that

\begin{equation*}
\begin{split}
  f(\Theta+1) & = max(g(\Theta+1),f(\Theta)) + h(\Theta+1) \\
   & = max(g(\Theta+1),B) + h(\Theta+1) \\
   f^\prime(\Theta+1) & = max(g(\Theta+1),B) + ((\Theta+1)-\Theta-1)\times max(g(\Theta+1),h(\Theta)) + h(\Theta+1) \\
   & = max(g(\Theta+1),B) + h(\Theta+1)
\end{split}
\end{equation*}

\noindent
\textbf{General Case.}  Assuming $f(x) \leq f^\prime(x)$, we need to
prove that $f(x+1) \leq f^\prime(x+1)$.
By induction hypothesis and Lemma~\ref{lemmamax2} we have that:

\begin{equation*}
\begin{split}
  f(x+1)  & = max(g(x+1),f(x)) + h(x+1) \\ 
  & \leq max(g(x+1), max(g(x),B) + (x - \Theta -1) \times max(g(x),h(x-1)) + h(x)) + h(x+1)
\end{split}
\end{equation*}

\noindent
By Lemma~\ref{lemmamax} we have that:

\begin{equation} \label{eq1}
\begin{split}
  f(x+1) & \leq max(g(x+1), max(g(x),B)) \\
  & + max(g(x+1),(x - \Theta - 1) \times max(g(x),h(x-1))) \\
  & + max(g(x+1), h(x)) \\
  & + h(x+1)
\end{split}
\end{equation}

\noindent
Consider now the first term appearing in the sum of the right hand
side of the inequality (\ref{eq1}).
Since $max$ is associative, and it holds that $\forall x: g(x+1) \geq
g(x)$ (which follows from the hypothesis of the theorem), we obtain:

\begin{equation} \label{eq2}
\begin{split}
  max(g(x+1), max(g(x),B)) & = max(max(g(x+1),g(x)),B) \\
  & = max(g(x+1),B)
\end{split}
\end{equation}

\noindent
We consider now the second term in (\ref{eq1}). By
Lemma~\ref{lemmamax} we obtain:

\begin{equation*}
\begin{split}
& max(g(x+1), (x-\Theta-1) \times max(g(x),h(x-1))) \\
& \leq (x-\Theta-1) \times max(g(x+1), max(g(x),h(x-1))) 
\end{split}
\end{equation*}

\noindent
As before, by associativity of $max$, this is equivalent to:

\begin{equation*}
\begin{split}
(x-\Theta-1) \times max(g(x+1),h(x-1))
\end{split}
\end{equation*}

\noindent
By Lemma~\ref{lemmamax2}, and $h(x-1) \leq h(x)$ (by hypothesis), we have that: 

\begin{equation} \label{eq3}
  \begin{split}    
    (x-\Theta-1) \times max(g(x+1),h(x))
\end{split}
\end{equation}

\noindent
Replacing the results of (\ref{eq2}) and (\ref{eq3}) in (\ref{eq1}):

\begin{equation*} 
\begin{split}
  f(x+1) & \leq max(g(x+1),B) \\
  & + (x - \Theta - 1) \times max(g(x+1),h(x)) \\
  & + max(g(x+1), h(x)) + h(x+1) \\
  & = max(g(x+1),B) \\
  & + (x - \Theta) \times max(g(x+1),h(x)) + h(x+1) \\
  & = f^\prime(x+1)
\end{split}
\end{equation*}
\noindent
$\therefore f(x+1) \leq f^\prime(x+1)$
\end{proof}

\label{lastpage}
\end{document}